\documentclass[]{article}
\usepackage{amsmath,epsfig,subfig,amssymb}
\usepackage{float}
\usepackage{algorithm}
\usepackage{algorithmic}
\usepackage{microtype}
\usepackage{amsthm}

\usepackage[numbers]{natbib}

\usepackage[usenames,dvipsnames,dvinames]{xcolor}

\usepackage{ifthen}
\newboolean{isdraft}
\setboolean{isdraft}{true} 
\ifthenelse{\boolean{isdraft}}{%
\newcommand{\myaddcomment}[3]{{\color{#1}{\ensuremath{\langle\!\!\langle}{\bf {#2} :} {#3}\ensuremath{\rangle\!\!\rangle}}}}
\newcommand{\rishabh}[1]{\myaddcomment{LimeGreen}{Rishabh}{#1}}
\newcommand{\JTR}[1]{\myaddcomment{LimeGreen}{Jeff\ensuremath{\rightarrow}Rishabh}{#1}}
\newcommand{\jeff}[1]{\myaddcomment{blue}{Jeff}{#1}}
\newcommand{\RTJ}[1]{\myaddcomment{blue}{Rishabh\ensuremath{\rightarrow}Jeff}{#1}}
\newcommand{\toboth}[1]{\myaddcomment{red}{Rishabh \& Jeff}{#1}}
} {
\newcommand{\rishabh}[1]{}
\newcommand{\JTR}[1]{}
\newcommand{\jeff}[1]{}
\newcommand{\RTJ}[1]{}
\newcommand{\toboth}[1]{}
}

\providecommand{\doextended}{true}
\newboolean{isextended}
\setboolean{isextended}{\doextended} 
\ifthenelse{\boolean{isextended}}{%
\newcommand{\extendedv}[1]{#1}
\newcommand{\notextendedv}[1]{}
}{
\newcommand{\extendedv}[1]{}
\newcommand{\notextendedv}[1]{#1}
}

\author{ {\bf Rishabh Iyer} \\
Dept. of Electrical Engineering\\  
University of Washington\\ 
Seattle, WA-98175, USA
\and 
{\bf Jeff Bilmes} \\ 
Dept. of Electrical Engineering \\  
University of Washington\\ 
Seattle, WA-98175, USA
}

\DeclareMathOperator*{\argmax}{argmax}
\DeclareMathOperator*{\argmin}{argmin}
\newtheorem{theorem}{Theorem}[section]
\newtheorem{lemma}{Lemma}[section]
\newtheorem{definition}{Definition}[section]

\newtheorem{corollary}[theorem]{Corollary}
\newcommand{\lovasz}{Lov\'asz}

\title{Algorithms for Approximate Minimization of the Difference Between
  Submodular Functions, with Applications\thanks{A shorter version of this appeared in Proc. Uncertainty in Artificial Intelligence (UAI), Catalina Islands, 2012}}
\begin{document}
\maketitle

\begin{abstract}


  We extend the work of~\citet{narasimhanbilmes} for minimizing set
  functions representable as a difference between submodular
  functions. Similar to~\cite{narasimhanbilmes}, our new algorithms
  are guaranteed to monotonically reduce the objective function at
  every step. We empirically and theoretically show that the
  per-iteration cost of our algorithms is much less
  than~\cite{narasimhanbilmes}, and our algorithms can be used to
  efficiently minimize a difference between submodular functions under
  various combinatorial constraints, a problem not previously
  addressed.
  We provide computational bounds and a hardness result on the
  multiplicative inapproximability of minimizing the difference
  between submodular functions. We show, however, that it is possible
  to give worst-case additive bounds by providing a polynomial time
  computable lower-bound on the minima.
%
  Finally we show how a number of machine learning problems can be
  modeled as minimizing the difference between submodular functions. 
  We experimentally show the validity of our algorithms by testing them
  on the problem of feature selection with submodular
  cost features.  \looseness-1
\end{abstract}

\section{Introduction}
\label{sec:introduction}

Discrete optimization is important to many areas of machine learning
and recently an ever growing number of problems have been shown to be
expressible as submodular function minimization or maximization
(e.g.,~\cite{kkt03,krause2005near, krause08robust, linacl, linbudget,
  lin2011-speech-corpus-creation}). The class of submodular functions
is indeed special since submodular function minimization is known to
be polynomial time, while submodular maximization, although NP
complete, admits constant factor approximation algorithms. Let $V =
\{1, 2, \cdots, n\}$ refer a ground set, then $f: 2^V \rightarrow
\mathbb{R}$ is said to be submodular if for sets $S, T \subseteq V$,
$f(S) + f(T) \geq f(S \cup T) + f(S \cap T)$
(see~\cite{fujishige2005submodular} for details on submodular,
supermodular, and modular functions). Submodular functions have a
diminishing returns property, wherein the gain of an element in the
context of bigger set is lesser than the gain of that element in the
context of a smaller subset. This property occurs naturally in many
applications in machine learning, computer vision, economics,
operations research, etc. \looseness-1



In this paper, we address the following problem. Given two submodular
functions $f$ and $g$, 
and define $v(X) \triangleq f(X) - g(X)$,
solve the following optimization problem:\looseness-1
\begin{equation} \label{probstat}
\min_{X \subseteq V} [ f(X) - g(X) ]
\equiv
\min_{X \subseteq V} [ v(X) ].
\end{equation}
A number of machine learning problems involve minimization over a
difference between submodular functions. The following are some
examples: \looseness-1
\begin{itemize}

\item \textbf{Sensor placement with submodular costs: } The problem of
  choosing sensor locations $A$ from a given set of possible locations $V$ can
  be modeled \cite{krause2005near,krause2008near}
  by maximizing 
  the mutual information between the chosen variables $A$ and the
  unchosen set $V \backslash A$ (i.e., $f(A) = I(X_A; X_{V \backslash
    A})$). Alternatively, we may wish to maximize the mutual
  information between a set of chosen sensors $X_A$ and a fixed
  quantity of interest $C$ (i.e., $f(A) = I(X_A ; C)$) under the
  assumption that the set of features $X_A$ are conditionally independent given $C$
  \cite{krause2005near}.  These objectives are submodular and thus the
  problem becomes maximizing a submodular function subject to a
  cardinality constraint.  Often, however, there are costs $c(A)$
  associated with the locations that naturally have a diminishing
  returns property. For example, there is typically a discount when
  purchasing sensors in bulk. Moreover, there may be diminished cost
  for placing a sensor in a particular location given placement in certain
  other locations (e.g., the additional equipment needed to install a
  sensor in, say, a precarious environment could be re-used for
  multiple sensor installations in like environments). Hence, along
  with maximizing mutual information, we also want to simultaneously
  minimize the cost and this problem can be addressed by
  minimizing the difference between submodular functions $f(A) -
  \lambda c(A)$ for tradeoff parameter $\lambda$. \looseness-1

\item \textbf{Discriminatively structured graphical models and neural
    computation: } An application suggested in~\cite{narasimhanbilmes}
  and the initial motivation for this problem is to optimize the EAR
  criterion to produce a discriminatively structured graphical model.
  EAR is basically a difference between two mutual information
  functions (i.e., a difference between submodular
  functions). \cite{narasimhanbilmes} shows how classifiers based on
  discriminative structure using EAR can significantly outperform
  classifiers based on generative graphical models.  Note also that
  the EAR measure is the same as ``synergy'' in a neural code
  \cite{brenner2000synergy}, widely used in neuroscience.

\item \textbf{Feature selection: } Given a set of features $X_1, X_2,
  \cdots, X_{|V|}$, the feature selection problem is to find a small
  subset of features $X_A$ that work well when used in a pattern
  classifier. This problem can be modeled as maximizing the mutual
  information $I(X_A; C)$ where $C$ is the class. Note that $I(X_A; C)
  = H(X_A) - H(X_A | C)$ is always a difference between submodular
  functions. Under the na\"{\i}ve Bayes model, this function is
  submodular~\cite{krause2005near}. It is not submodular under general
  classifier models such as support vector machines (SVMs) or neural
  networks.  Certain features, moreover, might be cheaper to use given
  that others are already being computed. For example, if a subset
  $S_i \subseteq V$ of the features for a particular information
  source $i$ are spectral in nature, then once a particular $v \in
  S_i$ is chosen, the remaining features $S_i \setminus \{ v \}$ may
  be relatively inexpensive to compute, due to grouped computational
  strategies such as the fast Fourier transform. Therefore, it might
  be more appropriate to use a submodular cost model $c(A)$.  One such
  cost model might be $c(A) = \sum_i \sqrt{ m(A \cap S_i) }$ where
  $m(j)$ would be the cost of computing feature $j$.  Another might be
  $c(A) = \sum_i c_i \min( |A \cap S_i|, 1 )$ where $c_i$ is the cost of
  source $i$. Both offer diminishing cost for choosing features from
  the same information source. Such a cost model could be useful even
  under the na\"{\i}ve Bayes model, where $I(X_A; C)$ is
  submodular. Feature selection becomes a problem of maximizing
  $I(X_A; C) - \lambda c(A) = H(X_A) - [ H(X_A|C) + \lambda c(A)]$,
  the difference between two submodular functions.\looseness-1

\item \textbf{Probabilistic Inference:} \extendedv{A typical instance of
  probabilistic inference is the following:} We are given a
  distribution $p(x) \propto \exp(-v(x))$ where $x \in \{0,1\}^n$ and
  $v$ is a pseudo-Boolean function \cite{Boros2002155}. It is
  desirable to compute $\argmax_{x \in \{0,1\}^n} p(x)$ which means
  minimizing $v(x)$ over $x$, the most-probable explanation (MPE)
  problem \cite{pearl88}.  If $p$ factors with respect to a graphical
  model of tree-width $k$, then $v(x) = \sum_i v_i(x_C)$ where $C_i$
  is a bundle of indices such that $|C| \leq k+1$ and the sets $\mathcal C = \{ C_i
  \}_i$ form a junction tree, and it might be possible to solve
  inference using dynamic programming. If $k$ is large and/or if
  hypertree factorization does not hold, then approximate inference is
  typically used \cite{wainwright2008graphical}. On the other hand,
  defining $x(X) = \{ x \in \{0,1\}^n : x_i = 1 \text{ whenever } i
  \in X \}$, if the set function $\bar v(X) = v(x(X))$ is submodular,
  then even if $p$ has large tree-width, the MPE problem can be solved
  exactly in polynomial time
  \cite{jegelka2011-inference-gen-graph-cuts}. This, in fact, is the
  basis behind inference in many computer vision models where $v$ is
  often not only submodular but also has limited sized $|C_i|$. For
  example, for submodular $v$ and if $|C_i| \leq 2$ then graph-cuts
  can solve the MPE problem extremely rapidly
  \cite{kolmogorov2004energy} and even some cases with $v$
  non-submodular \cite{kolmogorov2007minimizing}.  An important
  challenge is to consider non-submodular $v$ that can be minimized
  efficiently and for which there are approximation guarantees, a
  problem recently addressed in \cite{jegelka2011-nonsubmod-vision}.
  On the other hand, if $v$ can be expressed as a difference between
  two submodular functions (which it can, see Lemma~\ref{thm1}), or if
  such a decomposition can be computed (which it sometimes can, see
  Lemma~\ref{thm2}), then a procedure to minimize the difference
  between two submodular functions offers new ways to solve
  probabilistic inference. As an example, a large class of rich higher potentials can be expressed as~\cite{gallagher2011inference}:
  \begin{align}
f(x) = \sum_{C \in \mathcal C} w_C \prod_{i \in C} x_i  
  \end{align}
$\mathcal C$ here stands for a set of sets, possibly with higher-order terms (i.e there exist $C \in \mathcal C: |C| > 2$). If $w_C \leq 0, \forall C \in \mathcal C$, then $f$ is submodular. If $|\mathcal C|$ is not large (say polynomial in $n$), we can efficiently find a decomposition into submodular components (which will contain the sets $C \in \mathcal C: w_C \leq 0$) and the supermodular terms (which contain sets $C \in \mathcal C: w_C \geq 0$). These can potentially represent a rich class of potential functions for a number of applications, particularly in vision. 
\end{itemize} 

\extendedv{We note that given a solution to Equation~\ref{probstat},
  we can also minimize the difference between two supermodular
  functions $\min((-g) - (-f))$, maximize the difference between two
  submodular functions $\max (-v) = \max(g-f)$, and maximize the
  difference between two supermodular functions $\max (-v) =
  \max((-f)-(-g))$.}

Previously,~\citet{narasimhanbilmes} proposed an algorithm inspired by
the convex-concave procedure~\cite{yuille2002concave} to address
Equation~\eqref{probstat}.  This algorithm iteratively minimizes a
submodular function by replacing the second submodular function $g$ by
it's modular lower bound. They also show that any set function can be
expressed as a difference between two submodular functions and hence
every set function optimization problem can be reduced to minimizing a
difference between submodular functions. They show that this process
converges to a local minima, however the convergence rate is left as
an open question.  \looseness-1

In this paper, we first describe tight modular bounds on submodular
functions in Section~\ref{sec:modular-upper-lower}, including lower
bounds based on points in the base polytope as used
in~\cite{narasimhanbilmes}, and recent upper bounds first described in a 
result in~\cite{jegelkacvpr}.
%
In section~\ref{sec:subm-superm-proc}, we describe the
submodular-supermodular procedure proposed in~\cite{narasimhanbilmes}.
We further provide a constructive procedure for finding the submodular
functions $f$ and $g$ for any arbitrary set function $v$. Although our
construction is NP hard in general,
we show how for certain classes of set functions $v$, it is possible
to find the decompositions $f$ and $g$ in polynomial time. In
Section~\ref{sec:altern-algor-minim}, we propose two new algorithms
both of which are guaranteed to monotonically reduce the objective at
every iteration and which converge to a local minima. 
Further we note that the per-iteration cost of our algorithms
is in general much less than~\cite{narasimhanbilmes}, and
empirically verify that our algorithms are orders of magnitude faster
on real data. We show that, unlike in~\cite{narasimhanbilmes}, our
algorithms can be extended to easily optimize
equation~\eqref{probstat} under cardinality, knapsack, and matroid
constraints. Moreover, one of our algorithms can actually handle
complex combinatorial constraints, such as spanning trees, matchings,
cuts, etc. Further in Section~\ref{sec:theoretical-results}, we give a
hardness result that there does not exist any polynomial time
algorithm with any polynomial time multiplicative approximation
guarantees unless P=NP, even when it is easy to find or when we are
given the decomposition $f$ and $g$, thus justifying the need for
heuristic methods to solve this problem. We show, however, that it is
possible to get additive bounds by showing polynomial time computable
upper and lower bound on the optima. We also provide computational
bounds for all our algorithms (including the submodular-supermodular
procedure), a problem left open in~\cite{narasimhanbilmes}.
%


Finally we perform a number of experiments on the feature selection
problem under various cost models, and show how our algorithms used to
maximize the mutual information perform better than greedy
selection (which would be near optimal under the na\"{\i}ve Bayes
assumptions) and with less cost.



\section{Modular Upper and Lower bounds}
\label{sec:modular-upper-lower}

The Taylor series approximation of a convex function provides a
natural way of providing lower bounds on such a function. In
particular the first order Taylor series approximation of a convex function
is a lower bound on the function, and is linear in $x$ for a given $y$
and hence given a convex function $\phi$, we have: \looseness-1
\begin{equation}
 \phi(x) \geq \phi(y) + \langle \nabla \phi(y), x - y \rangle.
\end{equation}

Surprisingly, any submodular function has both a tight lower
\cite{edmondspolyhedra} {\bf and} upper bound \cite{jegelkacvpr}, unlike
strict convexity where there is only a tight first order lower
bound. \looseness-1

\subsection{Modular Lower Bounds}
\label{sec:modular-lower-bounds}

Recall that for submodular function $f$, the submodular polymatroid,
base polytope and the sub-differential with respect to a set
$Y$~\cite{fujishige2005submodular} are respectively:
\looseness-1
\begin{gather}
\mathcal P_f = \{ x : x(S) \leq f(S), \forall S \subseteq V \} \\
 \mathcal B_f = \mathcal P_f \cap \{ x : x(V) = f(V) \}  \\
\!\!\!\!\partial f(Y) = \{y \in \mathbb{R}^V : \forall X \subseteq V, f(Y) - y(Y) \leq f(X) - y(X)\} \notag
\end{gather}
The extreme points of this sub-differential are easy to find and
characterize, and can be obtained from a greedy algorithm
(\cite{edmondspolyhedra,fujishige2005submodular}) as follows:
\begin{theorem} (\cite{fujishige2005submodular}, Theorem 6.11)
A point $y$ is an extreme point of $\partial f(Y)$, \textit{iff} there exists a chain $\emptyset = S_0 \subset S_1 \subset \cdots \subset S_n$ with $Y = S_j$ for some $j$, such that $y(S_i \setminus S_{i-1}) = y(S_i) - y(S_{i-1}) = f(S_i) - f(S_{i-1})$.  \looseness-1
\end{theorem}

Let $\sigma$ be a permutation of $V$ and define $S_i^\sigma = \{ \sigma(1),
\sigma(2), \dots, \sigma(i) \}$ as $\sigma$'s chain containing $Y$,
meaning $S_{|Y|}^\sigma = Y$ (we say that $\sigma$'s chain
\underline{contains} $Y$). Then we can define a sub-gradient $h^f_Y$
corresponding to $f$ as:
\begin{equation}
\label{eq:permmod}
 h^f_{Y, \sigma}(\sigma(i)) = 
\begin{cases}
f(S_1^\sigma) & \text{ if } i=1 \\
f(S_i^\sigma) - f(S_{i-1}^\sigma) & \text { otherwise }
\end{cases}. \nonumber
\end{equation}
We get a modular lower bound of $f$ as follows:
\begin{align*}
h^f_{Y, \sigma}(X) \leq f(X), \forall X \subseteq V,
\text{ and } \forall i,
h^f_{Y, \sigma}(S_i^\sigma) = f(S_i^\sigma) ,
\end{align*}
which is parameterized by a set $Y$ and a permutation $\sigma$. Note
$h(X) = \sum_{i \in X} h(i)$, and $h^f_{Y, \sigma}(Y) = f(Y)$. Observe
the similarity to convex functions, where a linear lower bound is
parameterized by a vector $y$. \looseness-1

\subsection{Modular Upper Bounds}
\label{sec:modular-upper-bounds}

For $f$ submodular, \cite{nemhauser1978} established the following:
\looseness-1
\begin{eqnarray} 
\label{nembounds1} f(Y) \leq f(X) - \sum_{j \in X \backslash Y } f(j| X \backslash j) + \sum_{j \in Y \backslash X} f(j| X \cap Y),\nonumber \\
\label{nembounds2} f(Y) \leq f(X) - \sum_{j \in X \backslash Y } f(j| (X \! \cup \! Y) \backslash j) + \sum_{j \in Y \backslash X} f(j | X) \nonumber
\end{eqnarray}
Note that $f(A| B) \triangleq f(A \cup B) - f(B)$ is the gain of
adding $A$ in the context of $B$. These upper bounds in fact
characterize submodular functions, in that a function $f$ is a
submodular function \textit{iff} it follows either of the above
bounds. Using the above, two tight modular upper
bounds~(\cite{jegelkacvpr}) can be defined as follows:\looseness-1
\begin{gather}
\label{modnembounds1} 
\!\!\!f(Y) \leq m^f_{X, 1}(Y) \triangleq f(X) - \!\!\! \sum_{j \in X \backslash Y } f(j| X \backslash j) + \sum_{j \in Y \backslash X} f(j| \emptyset), \notag \\
\label{modnembounds2} 
\!\!\!f(Y) \leq m^f_{X, 2}(Y) \triangleq f(X) - \!\!\! \sum_{j \in X \backslash Y } f(j| V \backslash j) + \sum_{j \in Y \backslash X} f(j| X). \notag
\end{gather}
Hence, this yields two tight (at set $X$) modular upper bounds $m^f_{X,
  1}, m^f_{X, 2}$ for any submodular function $f$.  For briefness,
when referring either one we use $m^f_{X}$. \looseness-1

\section*{\addtocounter{section}{1}\thesection~Submodular-Supermodular~Procedure}
\label{sec:subm-superm-proc}

We now review the submodular-supermodular
procedure~\cite{narasimhanbilmes} to minimize functions expressible as
a difference between submodular functions (henceforth called DS
functions). Interestingly, any set function can be expressed as a DS
function using suitable submodular functions as shown below. The
result was first shown in~\cite{narasimhanbilmes} using the \lovasz{}
extension. We here give a new combinatorial proof, which avoids
Hessians of polyhedral convex functions and which provides a way of
constructing (a non-unique) pair of submodular functions $f$ and $g$
for an arbitrary set function $v$.\looseness-1
\begin{lemma}\cite{narasimhanbilmes}
\label{thm1}
Given any set function $v$, it can be expressed as a DS functions $v(X) = f(X) - g(X), \forall X \subseteq V$ for some submodular functions $f$ and $g$.
\looseness-1
\end{lemma}
\begin{proof}
  Given a set function $v$, we can define $\alpha = \min_{X \subset Y
    \subseteq V \setminus j} v(j|X) - v(j|Y)$\footnote{We denote
    $j,X,Y : X \subset Y \subseteq V \setminus \{ j \}$ by $X \subset
    Y \subseteq V \setminus j$.}. Clearly $\alpha < 0$, since
  otherwise $v$ would be submodular. Now consider any (strictly)
  submodular function $g$, i.e., one having $\beta = \min_{X \subset Y
    \subseteq V \setminus j} g(j|X) - g(j|Y) > 0$. Define $f'(X) =
  v(X) + \frac{|\alpha'|}{\beta}g(X)$ with any $\alpha' \leq \alpha$. Now
  it is easy to see that $f'$ is submodular since $\min_{X \subset Y
    \subseteq V \setminus j} f'(j|X) - f'(j|Y) \geq \alpha + |\alpha'| \geq
  0$. Hence $v(X) = f'(X) - \frac{|\alpha'|}{\beta}g(X)$, is a
  difference between two submodular functions. \looseness-1
\end{proof}
The above proof requires the computation of $\alpha$ and $\beta$ which
has, in general, exponential complexity.  Using the construction
above, however, it is easy to find the decomposition $f$ and $g$ under
certain conditions on $v$.
\begin{lemma} 
\label{thm2} 
If $\alpha$ or at least a lower bound on $\alpha$ for any set function
$v$ can be computed in polynomial time, functions $f$ and $g$
corresponding to $v$ can obtained in polynomial time.
\end{lemma}
\begin{proof}
  Define $g$ as $g(X) = \sqrt{|X|}$. Then
  $\beta 
  = \min_{X \subset Y \subseteq V \setminus j} \sqrt{|X| + 1} - \sqrt{|X|} -
  \sqrt{|Y| + 1} + \sqrt{|Y|} = \min_{X \subset V \setminus j} \sqrt{|X| + 1} -
  \sqrt{|X|} - \sqrt{|X| + 2} + \sqrt{|X| + 1} = 2\sqrt{n-1} -
  \sqrt{n} - \sqrt{n-2}$. The last inequality follows since 
  the smallest difference in gains will occur at $|X| = n-2$. Hence
  $\beta$ is easily computed, and given a lower bound on
  $\alpha$, from Lemma~\ref{thm1} the decomposition
  can be obtained in polynomial time. A similar argument holds for $g$
  being other concave functions over $|X|$. \looseness-1
\end{proof}

\begin{algorithm}[tb]
\caption{The submodular-supermodular (SubSup) procedure~\cite{narasimhanbilmes}}
\begin{algorithmic}[1]
\STATE $X^0 = \emptyset$ ; $t \gets 0$ ;
\WHILE{not converged (i.e., $(X^{t+1} \neq X^t)$)}
\STATE Randomly choose a permutation $\sigma^t$ whose chain contains the set $X^t$.
\STATE $X^{t+1}:= \argmin_{X} f(X) - h^g_{X^t,\sigma^t}(X) $
\STATE $t \leftarrow t+1$
\ENDWHILE
\end{algorithmic}
\label{alg:ssp}
\end{algorithm}
The submodular supermodular (SubSup) procedure is given in
Algorithm~\ref{alg:ssp}.  At every step of the algorithm, we minimize
a submodular function which can be performed in strongly polynomial
time~\cite{orlin2009faster, schrijver2000combinatorial}
although the best known complexity is $O(n^5 \eta + n^6)$ where $\eta$
is the cost of a function evaluation.
Algorithm~\ref{alg:ssp}
is guaranteed to converge to a local minima and moreover the
algorithm monotonically decreases the function objective at every
iteration, as we show below.
\begin{lemma}\cite{narasimhanbilmes} \looseness-1
Algorithm~\ref{alg:ssp} is guaranteed to decrease the objective function at every iteration. Further, the algorithm is guaranteed to converge to a local minima by checking at most $O(n)$ permutations at every iteration.
\end{lemma}
\notextendedv{Due to space constraints, we omit the proof of this lemma which is in any case described in~\cite{narasimhanbilmes, extended}.}
\extendedv{\begin{proof}
The objective reduces at every
  iteration since:
\begin{align}
f(X^{t+1}) - g(X^{t+1}) &\overset{a}{\leq}& f(X^{t+1}) - h^g_{X^t,\sigma^t}(X^{t+1}) \nonumber \\	
			&\overset{b}{\leq}& f(X^t) - h^g_{X^t,\sigma^t}(X^{t}) \nonumber \\
			&\overset{c}{=}& f(X^t) - g(X^{t}) \nonumber
\end{align}
Where (a) follows since $h^g_{X^t,\sigma^t}(X^{t+1}) \leq g(X^{t+1})$,
and (b) follows since $X^{t+1}$ is the minimizer of $f(X) -
h^g_{X^t,\sigma^t}(X)$, and (c) follows since
$h^g_{X^t,\sigma^t}(X^{t}) = g(X^t)$ from the tightness of the modular
lower bound.

Further note that, if there is no improvement in the function value by considering $O(n)$ permutations 
each with different elements at $\sigma^t(|X^t| - 1)$ and $\sigma^t(|X^t| + 1)$,  
then this is equivalent to a local minima condition on $v$ since $h^g_{X^t,\sigma^t}(S_{|X^t| + 1}^\sigma) = f(S_{|X^t| + 1}^\sigma)$ 
and $h^g_{X^t,\sigma^t}(S_{|X^t| - 1}^\sigma) = f(S_{|X^t| - 1}^\sigma)$.
\end{proof}}

Algorithm~\ref{alg:ssp} requires performing a submodular function
minimization at every iteration which while polynomial in $n$ is (due
to the complexity described above) not practical for large problem
sizes.  So while the algorithm reaches a local minima, it can be
costly to find it. A desirable result, therefore, would be to develop
new algorithms for minimizing DS functions, where the new
algorithms have the same properties as the SubSup procedure but are
much faster in practice.  We give this in the following sections.

\section{Alternate algorithms for minimizing DS functions} 
\label{sec:altern-algor-minim}

In this section we propose two new algorithms to minimize DS
functions, both of which are guaranteed to monotonically reduce the
objective at every iteration and converge to local minima. We briefly
describe these algorithms in the subsections below.

\subsection{The supermodular-submodular (SupSub) procedure}
\label{sec:superm-subm-sups}

In the submodular-supermodular procedure we iteratively minimized
$f(X) - g(X)$ by replacing $g$ by it's modular lower bound at every
iteration. We can instead replace $f$ by it's modular upper bound as
is done in Algorithm~\ref{alg:supsub}, which leads to the {\em
  supermodular-submodular} procedure.

\begin{algorithm}[h]
\caption{The supermodular-submodular (SupSub) procedure}
\begin{algorithmic}[1]
\STATE $X^0 = \emptyset$ ; $t \gets 0$ ;
\WHILE{not converged (i.e., $(X^{t+1} \neq X^t)$)}
\STATE $X^{t+1}:= \argmin_X m^f_{X^t}(X) - g(X)$ \label{line:submaxline}
\STATE $t \leftarrow t+1$
\ENDWHILE
\end{algorithmic}
\label{alg:supsub}
\end{algorithm}

In the SupSub procedure, at every step we perform submodular
maximization which, although NP complete to solve exactly, admits a
number of fast constant factor approximation
algorithms~\cite{feldman2012optimal, fiege2011submodmax}. Notice that
we have two modular upper bounds and hence there are a number of ways
we can choose between them. One way is to run both maximization
procedures with the two modular upper bounds at every iteration in
parallel, and choose the one which is better. Here by better we mean
the one in which the function value is lesser. Alternatively we can
alternate between the two modular upper bounds by first maximizing the
expression using the first modular upper bound, and then maximize the
expression using the second modular upper bound. Notice that since we
perform approximate submodular maximization at every iteration, we are
not guaranteed to monotonically reduce the objective value at every
iteration. If, however, we ensure that at every iteration we take the
next step only if the objective $v$ does not increase, we will restore
monotonicity at every iteration. Also, in some cases we converge to
local optima as shown in the following theorem.
\begin{theorem} \looseness-1 Both variants of the
  supermodular-submodular procedure (Algorithm~\ref{alg:supsub})
  monotonically reduces the objective value at every
  iteration. Moreover, assuming a submodular maximization
  procedure in line~\ref{line:submaxline} that 
  reaches a local maxima of $m^f_{X^t}(X) - g(X)$, then
  if Algorithm~\ref{alg:supsub} does not
  improve under both modular upper bounds then it reaches
  a local optima of $v$.
\end{theorem}
\begin{proof}
For either modular upper bound, we have:
\begin{align}
f(X^{t+1}) - g(X^{t+1}) 
&\overset{a}{\leq} m^f_{X^t}(X^{t+1}) -  g(X^{t+1}) \nonumber \\	
&\overset{b}{\leq} m^f_{X^t}(X^{t}) -  g(X^{t}) \nonumber \\
&\overset{c}{=} f(X^t) - g(X^{t}), \nonumber
\end{align}
where (a) follows since $f(X^{t+1}) \leq m^f_{X^t}(X^{t+1})$, and
(b) follows since we assume that we take the next step only if the objective value does not increase and (c) follows since $m^f_{X^t}(X^{t}) = f(X^t)$ from the
tightness of the modular upper bound.

To show that this algorithm converges to a local minima, we assume that the submodular maximization
  procedure in line~\ref{line:submaxline} converges to a local maxima. Then observe that
if the objective value does not decrease in an iteration under
both upper bounds, it implies
that $m^f_{X^t}(X^t) - g(X^t)$ is already a local optimum
in that (for both upper bounds) we have $m^f_{X^t}(X^t \cup j) - g(X^t \cup j)
\geq m^f_{X^t}(X^t) - g(X^t), \forall j \notin X^t$ and $m^f_{X^t}(X^t
\backslash j) - g(X^t \backslash j) \geq m^f_{X^t}(X^t) - g(X^t),
\forall j \in X^t$.
Note that
$m^f_{X^t, 1}(X^t
\backslash j) = f(X^t) - f(j | X^t \backslash j) = f(X^t \backslash j)$ and $m^f_{X^t, 2}(X^t \cup j) = f(X^t) + f(j | X^t) = f(X^t \cup j)$  
and hence if both modular upper bounds are at a local optima, it implies 
$f(X^t) - g(X^t) = m^f_{X^t, 1}(X^t) - g(X^t) \leq m^f_{X^t, 1}(X^t \backslash j) - g(X^t \backslash j) = f(X^t \backslash j) - g(X^t \backslash j)$. Similarly $f(X^t) - g(X^t) = m^f_{X^t, 2}(X^t) - g(X^t) \leq m^f_{X^t, 2}(X^t \cup j) - g(X^t \cup j) = f(X^t \cup j) - g(X^t \cup j)$.
Hence $X^t$ is a local optima for $v(X) = f(X) - g(X)$, since $v(X^t) \leq v(X^t \cup j)$ and $v(X^t) \leq v(X^t \backslash j)$.
\end{proof}

To ensure that we take the largest step at each iteration, we can use
the recently proposed tight (1/2)-approximation algorithm
in~\cite{feldman2012optimal} for unconstrained non-monotone submodular
function maximization --- this is the best possible in polynomial time
for the class of submodular functions independent of the P=NP
question. The algorithm is a form of bi-directional randomized greedy
procedure and, most importantly for practical considerations, is
linear time~\cite{feldman2012optimal}. In practice we just use a combination of a form of a simple greedy procedure, and the bi-directional randomized algorithm, by picking the best amongst the two at every iteration. Since the randomized greedy algorithm is $1/2$ approximate, the combination of the two procedures also will be $1/2$ approximate.

Lastly, note that this algorithm is closely related to a local search
heuristic for submodular maximization \cite{fiege2011submodmax}. In
particular, if instead of using the greedy algorithm entirely at every
iteration, we take only one local step, we get a local search
heuristic. Hence, via the SupSub procedure, we may take larger steps
at every iteration as compared to a local search heuristic.

\subsection{The modular-modular (ModMod) procedure}
\label{sec:modul-modul-proc}

The submodular-supermodular procedure and the supermodular-submodular
procedure were obtained by replacing $g$ by it's modular lower bound
and $f$ by it's modular upper bound respectively. We can however
replace both of them by their respective modular bounds, as is done in
Algorithm~\ref{alg:modmod}.
\begin{algorithm}[h]
\caption{Modular-Modular (ModMod) procedure}
\begin{algorithmic}[1]
\STATE $X^0 = \emptyset$; $t \gets 0$ ;
\WHILE{not converged (i.e., $(X^{t+1} \neq X^t)$)}
\STATE Choose a permutation $\sigma^t$ whose chain contains the set $X^t$.
\STATE $X^{t+1}:= \argmin_X  m^f_{X^t}(X) - h^g_{X^t,\sigma^t}(X)$
\STATE $t \leftarrow t+1$
\ENDWHILE
\end{algorithmic}
\label{alg:modmod}
\end{algorithm}

In this algorithm at every iteration we minimize only a modular
function which can be done in $O(n)$ time, so this is extremely easy
(i.e., select all negative elements for the smallest minimum, or all
non-positive elements for the largest minimum). Like before, since we
have two modular upper bounds, we can use any of the variants
discussed in the subsection above.  Moreover, we are still guaranteed
to monotonically decrease the objective at every iteration and
converge to a local minima.\looseness-1
\begin{theorem}
  Algorithm~\ref{alg:modmod} monotonically decreases the function
  value at every iteration. If the function value does not
  increase on checking $O(n)$ different permutations with different
  elements at adjacent positions and with both modular upper bounds,
  then we have reached a local minima of $v$.
\end{theorem}
\begin{proof}
Again we can use similar reasoning as the earlier proofs and observe that:
\begin{align}
f(X^{t+1}) - g(X^{t+1}) 
&\leq m^f_{X^t}(X^{t+1}) -  h^g_{X^t,\sigma^t}(X^{t+1}) \nonumber \\	
&\leq m^f_{X^t}(X^t) -  h^g_{X^t,\sigma^t}(X^t) \nonumber \\
&= f(X^t) - g(X^t) \nonumber
\end{align}

We see that considering $O(n)$ permutations each with different
elements at $\sigma^t(|X^t| - 1)$ and $\sigma^t(|X^t| + 1)$, we
essentially consider all choices of $g(X^t \cup j)$ and $g(X^t
\backslash j)$, since $h^g_{X^t,\sigma^t}(S_{|X^t| + 1}) = f(S_{|X^t|
  + 1})$ and $h^g_{X^t,\sigma^t}(S_{|X^t| - 1}) = f(S_{|X^t| -
  1})$. Since we consider both modular upper bounds, we
correspondingly consider every choice of $f(X^t \cup j)$ and $f(X^t
\backslash j)$. Note that at convergence we have that $m^f_{X^t}(X^t)
- h^g_{X^t,\sigma^t}(X^t) \leq m^f_{X^t}(X) - h^g_{X^t,\sigma^t}(X),
\forall X \subseteq V$ for $O(n)$ different permutations and both
modular upper bounds. Correspondingly we are guaranteed that (since
the expression is modular) $\forall j \notin X^t, v(j| X^t) \geq 0$
and $\forall j \in X^t, v(j| X^t \backslash j) \geq 0$, where $v(X) =
f(X) - g(X)$. Hence the algorithm converges to a local minima.
\end{proof}

An important question is the choice of the permutation $\sigma^t$ at
every iteration $X^t$. We observe experimentally that the quality of
the algorithm depends strongly on the choice of permutation. Observe
that $f(X) - g(X) \leq m^f_{X^t}(X) - h^g_{X^t,\sigma^t}(X)$, and
$f(X^t) - g(X^t) = m^f_{X^t}(X^t) - h^g_{X^t,\sigma^t}(X^t)$. Hence,
we might obtain the greatest local reduction in the value of $v$ by
choosing permutation $\sigma^* \in \argmin_{\sigma} \min_X
(m^f_{X^t}(X) - h^g_{X^t,\sigma^t}(X))$, or the one which maximizes
$h^g_{X^t,\sigma^t}(X)$. We in fact might expect that choosing
$\sigma^t$ ordered according to greatest gains of $g$, with respect to
$X^t$, we would achieve greater descent at every iteration. Another
choice is to choose the permutation $\sigma$ based on the ordering of
gains of $v$ (or even $m^f_{X^t}$).  Through the former we are
guaranteed to at least progress as much as the local search
heuristic. Indeed, we observe in practice that the first two of these
heuristics performs much better than a random permutation for both the
ModMod and the SubSup procedure, thus addressing a question raised
in~\cite{narasimhanbilmes} about which ordering to use.  Practically
for the feature selection problem, the second heuristic seems to work
the best. \looseness-1

\subsection{Constrained minimization of a difference between submodular functions}

In this section we consider the problem of minimizing the difference
between submodular functions subject to constraints. We first note
that the problem of minimizing a submodular function under even simple
cardinality constraints in NP hard and also hard to
approximate~\cite{svitkina2008submodular}. Since there does not yet
seem to be a reasonable algorithm for constrained submodular
minimization at every iteration, it is unclear how we would use
Algorithm~\ref{alg:ssp}. However the problem of submodular
maximization under cardinality, matroid, and knapsack constraints
though NP hard admits a number of constant factor approximation
algorithms~\cite{nemhauser1978, lee2009non} and correspondingly the
cardinality constraints can be easily introduced in
Algorithm~\ref{alg:supsub}. Moreover, since a non-negative modular
function can be easily, directly and even exactly optimized under
cardinality, knapsack and matroid constraints
\cite{jegelka2010online}, Algorithm~\ref{alg:modmod} can also easily
be utilized. In addition, since problems such as finding the minimum
weight spanning tree, min-cut in a graph, etc., are polynomial time
algorithms in a number of cases, Algorithm~\ref{alg:modmod} can be
used when minimizing a non-negative function $v$ expressible as a
difference between submodular functions under combinatorial
constraints. If $v$ is non-negative, then so is its modular upper
bound, and then the ModMod procedure can directly be used for this
problem --- each iteration minimizes a non-negative modular function
subject to combinatorial constraints which is easy in many
cases~\cite{jegelka2010online, jegelka2010cooperative}. \looseness-1

\section{Theoretical results} 
\label{sec:theoretical-results}

In this section we analyze the computational and approximation
bounds for this problem. For simplicity we assume that the function
$v$ is normalized, i.e $v(\emptyset) = 0$. Hence we assume that $v$
achieves it minima at a negative value and correspondingly the
approximation factor in this case will be less than $1$.

We note in passing that the results in this section are mostly
negative, in that they demonstrate theoretically how complex a general
problem such as $\min_X [f(X) - g(X)]$ is, even for submodular $f$ and
$g$. In this paper, rather than consider these hardness results
pessimistically, we think of them as providing justification for the
heuristic procedures given in Section~\ref{sec:altern-algor-minim}
and \cite{narasimhanbilmes}.  In many cases, inspired heuristics can
yield 
good quality 
and hence practically useful
algorithms 
for real-world
problems. For example, the ModMod procedure
(Algorithm~\ref{alg:modmod}) and even the SupSub procedure
(Algorithm~\ref{alg:supsub}) can scale to very large problem sizes,
and thus can provide useful new strategies for the applications listed
in Section~\ref{sec:introduction}.

\subsection{Hardness}

Observe that the class of DS functions is essentially the class of
general set functions, and hence the problem of finding optimal
solutions is NP-hard. This is not surprising since general set
function minimization is inapproximable and there exist a large class
of functions where all (adaptive, possibly randomized) algorithms
perform arbitrarily poorly in polynomial time
\cite{trevisan2004inapproximability}. Clearly as is evident from
Theorem~\ref{thm1}, even the problem of finding the submodular
functions $f$ and $g$ requires exponential complexity. We moreover
show in the following theorem, however, that this problem is
multiplicatively inapproximable even when the functions $f$ and $g$
are easy to find. \looseness-1

\begin{theorem}\label{pnphardness}
  Unless P = NP, there cannot exist any polynomial time approximation
  algorithm for $\min_X v(X)$ where $v(X) = [f(X) - g(X)]$ is a positive set function and $f$ and $g$ are given
  submodular functions. In particular, let $n$ be the size of the
  problem instance, and $\alpha(n) > 0$ be any positive polynomial time
  computable function of $n$. If there exists a polynomial-time algorithm which
  is guaranteed to find a set $X': f(X') - g(X') < \alpha(n)
  \mbox{OPT}$, where OPT=$\min_X f(X) - g(X)$, then P = NP.
\end{theorem}  \looseness-1
\notextendedv{The proof of this theorem is in~\cite{extended}.}
\extendedv{
\begin{proof}
  We prove this by reducing this to the \textit{subset sum} problem.
  \extendedv{Given a positive modular function $m$ and a positive
    constant $t$, is there a subset $S \subseteq V$ such that $m(S) =
    t$?}  First we choose a random set $C$ (unknown to the algorithm),
  and define $t = m(C)$. Define a set function $v$, such that $v(S) =
  1$, if $m(S) = t$ and $v(S) = \frac{1}{\alpha(n)} - o(1)$
  otherwise. Observe that $\min_S v(S) = \frac{1}{\alpha(n)} - o(1)$, since
  $\alpha(n) > 1$. Note that $\alpha = \min_{X \subset Y \subseteq V
    \setminus j } v(j | X) - v(j | Y ) \geq 2(\frac{1}{\alpha(n)} -
  1)$. Hence we can easily compute a lower bound on $\alpha$ and hence
  from lemma~\ref{thm2} we can directly compute the decomposition $f$
  and $g$. In fact notice that the decomposition is directly
  computable since both $\alpha$ and $\beta$ are known.
  
  Now suppose there exists a polynomial time algorithm for this
  problem with an approximation factor of $\alpha(n)$. This implies that
  the algorithm is guaranteed to find a set $S$, such that $v(S) <
  1$. Hence this algorithm will solve the subset sum problem in
  polynomial time, which is a contradiction unless P = NP.
\end{proof}} 
In fact we show below that independent of the $P = NP$ question, there
cannot exist a sub-exponential time algorithm for this problem with any constant factor approximation. 
The theorem below gives information theoretic hardness for this
problem.
\begin{theorem}
  For any $0 < \epsilon < 1$, there cannot exist any deterministic (or
  possibly randomized) algorithm for $\min_X [f(X) - g(X)]$ (where $f$
  and $g$ are given submodular functions), that always finds a
  solution which is at most $\frac{1}{\epsilon}$ times the optimal, in fewer
  than $e^{\epsilon^2 n/8}$ queries.
\end{theorem}
\notextendedv{Again the proof of this theorem is in~\cite{extended}.}
\extendedv{\begin{proof}
For showing this theorem, we use the same proof technique as
in~\cite{fiege2011submodmax}. Define two sets $C$ and $D$, such
that $V = C \cup D$ and $|C| = |D| = n/2$. We then define a set
function $v(S)$ which depends only on $k = |S \cap C|$ and $l = |S
\cap D|$. 
In particular define $v(S) = \frac{1}{\epsilon}, \mbox{ if } |k - l|
\leq \epsilon n$ and $v(S) = 1, \mbox{ if } |k
- l| > \epsilon n$. Again, we have a trivial bound on $\alpha$
here
since 
$v(j | X) \geq \frac{1}{\epsilon} - 1$ and 
$v(j | Y) \leq 1 - \frac{1}{\epsilon}$. Hence, 
$\alpha = \min_{X \subset Y
    \subseteq V \setminus j} v(j|X) - v(j|Y) >2 |1 -
\frac{1}{\epsilon}|$. 
Thus, for this set
function, a decomposition $v=f-g$ can easily be obtained (Lemma~\ref{thm2}). \looseness-1

Now, let the partition $(C,D)$ be taken uniformly at random and
unknown to the algorithm. The algorithm issues some queries $S$ to the
value oracle. Call $S$ ``unbalanced'' if $|S \cap C|$ differs from $|S
\cap D|$ by more than $\epsilon n$. Recall the Chernoff
bounds~\cite{chernoff}:
Let $Y_1, Y_2, \cdots, Y_t$ be independent random variables in $[-1, 1]$, such that $\mathbb{E}[Y_i] = 0$, then:
\begin{equation}
Pr[\sum_{i = 1}^t Y_i > \lambda] \leq 2e^{-\lambda^2/2t}.
\end{equation}
Define $Y_i = I(i \in S)[ I(i \in C) - I(i \in D)]$. Clearly $Y_i \in [-1, 1]$, and we can use the bounds above. 
Hence for any query $S$, the
probability that $S$ is unbalanced is at most $2e^{-\epsilon^2 n/2}$. 
Thus, we can see that even
after $e^{\epsilon^2 n/4}$ number of queries, the probability that
the resulting set is unbalanced is still $2e^{-\epsilon^2 n/4}$. 
Hence
any algorithm will query only balanced sets 
regardless of $C$ and $D$,
and consequently with high probability the algorithm will obtain
$\frac{1}{\epsilon}$ as the minimum, while the actual minimum is
$1$. Thus, such an algorithm will never be able to achieve an
approximation factor better than $\frac{1}{\epsilon}$.
\end{proof}}

Essentially the theorems above say that even when we are given
(or can easily find) a decomposition such that $v(X) = f(X) - g(X)$,
there exist set functions such that any algorithm (either adaptive or
randomized) cannot be approximable upto any constant factor. It is possible that one could come up with an information theoretic construction to show this same result for any polynomial approximation factor. However under the assumption of P$\neq$NP, Theorem~\ref{pnphardness} shows that this problem is inapproximable upto any polynomial factor. Hence any algorithm trying to find the global optimum
for this problem~\cite{byrnes2009maximizing, kawahara2011prismatic} can only be exponential in the worst case.

Interestingly, the hardness results above holds even when the submodular functions $f$ and $g$ are monotone. This follows from the following Lemma:
\begin{lemma}
Given (not necessarily monotone) submodular functions $f$ and $g$, there exists monotone submodular functions $f^{\prime}$ and $g^{\prime}$ such that,
\begin{align}
f(X) - g(X) = f^{\prime}(X) - g^{\prime}(X), \forall X \subseteq V
\end{align}
\end{lemma}
\begin{proof}
The proof of this Lemma follows from a simple observation. The decomposition theorem of~\cite{cun82} shows that any submodular
function can be decomposed into a modular function plus a monotone
non-decreasing and \textit{totally normalized} polymatroid rank
function. Specifically, given submodular $f, g$ we have
\begin{align}
f^{\prime}(X) \triangleq f(X) - \sum_{j \in X} f(j | V \backslash j)
\end{align}
 and
\begin{align}
g^{\prime}(X) \triangleq g(X) - \sum_{j \in X} g(j | V \backslash j)
\end{align}
 $f^{\prime}, g^{\prime}$ are then totally normalized polymatroid rank
functions. Hence we have: $v(X) = f^{\prime}(X) - g^{\prime}(X) +
k(X)$, with modular $k(X) = \sum_{j \in X} v(j | V \backslash j)$. The idea is then to add $v(j)$ to $f^{\prime}$ if $v(j) \geq 0$ or add it to $g^{\prime}$ other-wise. In particular, let $V^+ = \{j: v(j) \geq 0\}$ and $V^- = \{j: v(j) < 0\}$. Notice that $V^+ \cup V^- = V$. Then,
\begin{align}
v(X) = f^{\prime}(X) + k(X \cap V^+) - \{g^{\prime} -  k(X \cap V^-)\}
\end{align}
Notice above that $f^{\prime}(X) + k(X \cap V^+)$ and $g^{\prime} -  k(X \cap V^-)$ are both monotone non-decreasing. Hence proved.
\end{proof}

This then implies the following corollary.
\begin{corollary}
Given submodular functions $f$ and $g$ such that $v(X) = f(X) - g(X) \geq 0$, the problem $\min_{X \subseteq V} v(X)$ is inapproximable, even if both $f$ and $g$ are monotone non-decreasing submodular.
\end{corollary}

\subsection{Polynomial time lower and upper bounds}

Since any submodular
function can be decomposed into a modular function plus a monotone
non-decreasing and \textit{totally normalized} polymatroid rank
function~\cite{cun82}, we have: $v(X) = f^{\prime}(X) - g^{\prime}(X) +
k(X)$, with modular $k(X) = \sum_{j \in X} v(j | V \backslash j)$ and $f^{\prime}$ and $g^{\prime}$ being the totally normalized polymatroid functions.\looseness-1
 
The algorithms in the previous sections are all based on repeatedly
finding upper bounds for $v$.  The following lower bounds directly
follow from the results above. \notextendedv{(The proof of this is
  in~\cite{extended})}\looseness-1
\begin{theorem} \label{thmlowbound}
We have the following two lower bounds on the minimizers of $v(X) = f(X) - g(X)$:
\begin{align}
\min_X v(X) &\geq \min_X f^{\prime}(X) + k(X) - g^{\prime}(V) \nonumber \\
\min_X v(X) &\geq f^{\prime}(\emptyset) - g^{\prime}(V) + \sum_{j \in V} \min(k(j), 0) \nonumber
\end{align} 
\end{theorem}%
\extendedv{\begin{proof}
Notice that 
\begin{align}
\min_X f(X) - g(X) 
 &= \min_X f^{\prime}(X) - g^{\prime}(X) + k(X) \nonumber \\
 &\geq \min_X ( f^{\prime}(X) + k(X)) - \max_X g^{\prime}(X)  \nonumber \\
 &= \min_X f^{\prime}(X) + k(X) - g^{\prime}(V) \nonumber
\end{align}
To get the second result, we start from the bound above and loosen it as:
\begin{align}
& \min_X f^{\prime}(X) + k(X) - g^{\prime}(V) \nonumber \\
&\geq \min_X f^{\prime}(X) + \min_X k(X) - g^{\prime}(V) \nonumber\\
&= f^{\prime}(\emptyset) + \sum_{j \in V} \min(v(j | V \backslash j), 0) - g^{\prime}(V) \nonumber \\
&= f^{\prime}(\emptyset) + \sum_{j \in V} \min(k(j), 0) - g^{\prime}(V)
\end{align}%
\vspace{-1ex}
\end{proof}
}

\vspace{-2ex}
The above lower bounds essentially provide 
bounds on the minima of the objective and thus can be used
to obtain an additive approximation guarantee. The algorithms
described in this paper are all polynomial time algorithms (as we show
below) and correspondingly from the bounds above we can get an 
estimate on how far we are from the optimal.

\subsection{Computational Bounds}
\label{sec:computational-bounds}

We now provide computational bounds for $\epsilon$-approximate versions
of our algorithms. 
Note that this
was left as an open question in~\cite{narasimhanbilmes}. 
Finding the
local minimizer of DS functions is PLS complete since it generalizes
the problem of finding the local optimum of the MAX-CUT
problem~\cite{schaffer1991simple}. 
\extendedv{Note that this trivially generalizes the MAX-CUT problem since if we set $f(X) = 0$ and $g(X)$ is the cut function, we get the max cut problem.
}%
However we show that an
$\epsilon$-approximate version of this algorithm will converge in
polynomial time.
\begin{definition}
  An $\epsilon$-approximate version of an iterative monotone
  non-decreasing algorithm for minimizing a set function $v$ is
  defined as a version of that algorithm, where we proceed to step
  $t+1$ only if $v(X^{t+1}) \leq v(X^t) (1 + \epsilon)$.
\end{definition}
Note that the $\epsilon$-approximate versions of
algorithms~\ref{alg:ssp}, \ref{alg:supsub} and~\ref{alg:modmod}, are
guaranteed to converge to $\epsilon$-approximate local optima. 
\extendedv{An $\epsilon$-approximate local optima of a function $v$ is a set $X$, such that $v(X \cup j) \geq v(X) (1 + \epsilon)$ and $v(X \backslash j) \geq v(X) (1 + \epsilon)$.
}%
W.l.o.g., assume 
that $X^0 = \emptyset$.
Then we have the
following computational bounds: \looseness-1
\begin{theorem}
  The $\epsilon$-approximate versions of algorithms~\ref{alg:ssp},
  \ref{alg:supsub} and~\ref{alg:modmod} have a worst case complexity
  of $O(\frac{\log(|M|/|m|)}{\epsilon}T))$, where $M =
  f^{\prime}(\emptyset) + \sum_{j \in V} \min(v(j | V \backslash j),
  0) - g^{\prime}(V)$, $m = v(X^1)$ and $O(T)$ is the complexity
  of every iteration of the algorithm (which corresponds to
  respectively the submodular minimization, maximization, or modular
  minimization in algorithms~\ref{alg:ssp}, \ref{alg:supsub}
  and~\ref{alg:modmod})..
\end{theorem}
\notextendedv{The proof of this theorem is in~\cite{extended}}
\extendedv{\begin{proof}
  Observe that $m = v(X^1) \leq v(X^0) = 0$. Correspondingly if $v(X^1) = 0$, it implies that the algorithm has converged, and cannot improve (since we are assuming our algorithms are $\epsilon-$approximate.  Hence in this case the algorithm will converge in one iteration. Consider then the case of $m < 0$.
  Note also from
  Theorem~\ref{thmlowbound} that $M = f^{\prime}(\emptyset) + \sum_{j
    \in V} \min(v(j | V \backslash j), 0) - g^{\prime}(V) < 0$ and
  that $\min_X f(X) - g(X) \geq M$. Since we are guaranteed to improve
  by a factor by at least $1 + \epsilon$ at every iteration we have
  that in $k$ iterations: $|m|(1 + \epsilon)^k \leq |M| \Rightarrow k
  = O(\frac{\log(|M|/|m|)}{\epsilon})$. 
Also since we assume that
  the complexity at every iteration is $O(T)$ we get the above result.
\end{proof}}

Observe that for the algorithms we use, $O(T)$ is strongly polynomial
in $n$. The best strongly polynomial time algorithm for submodular
function minimization is $O(n^5 \eta + n^6)$~\cite{orlin2009faster}
(the lower bound is currently unknown). Further the worst case
complexity of the greedy algorithm for maximization is $O(n^2)$ while
the complexity of modular minimization is just $O(n)$. Note finally
that these are worst case complexities and actually the algorithms run
much faster in practice.

\begin{figure}[h]
  \centering
  \subfloat[SVM]{\label{fig:mushroomssvm}\includegraphics[width=0.23\textwidth]{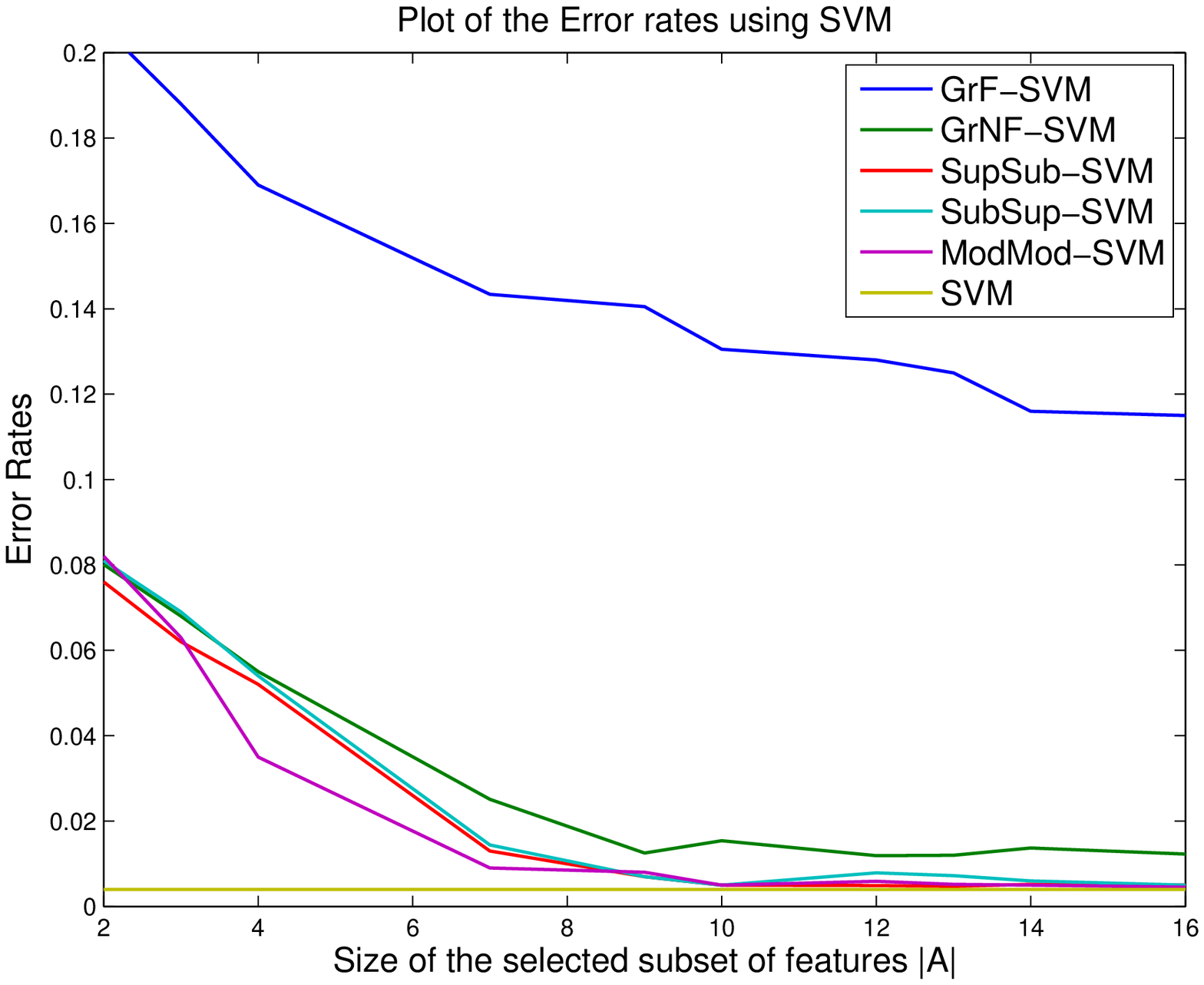}}
  ~ 
  \subfloat[NB]{\label{fig:mushroomsnb}\includegraphics[width=0.23\textwidth]{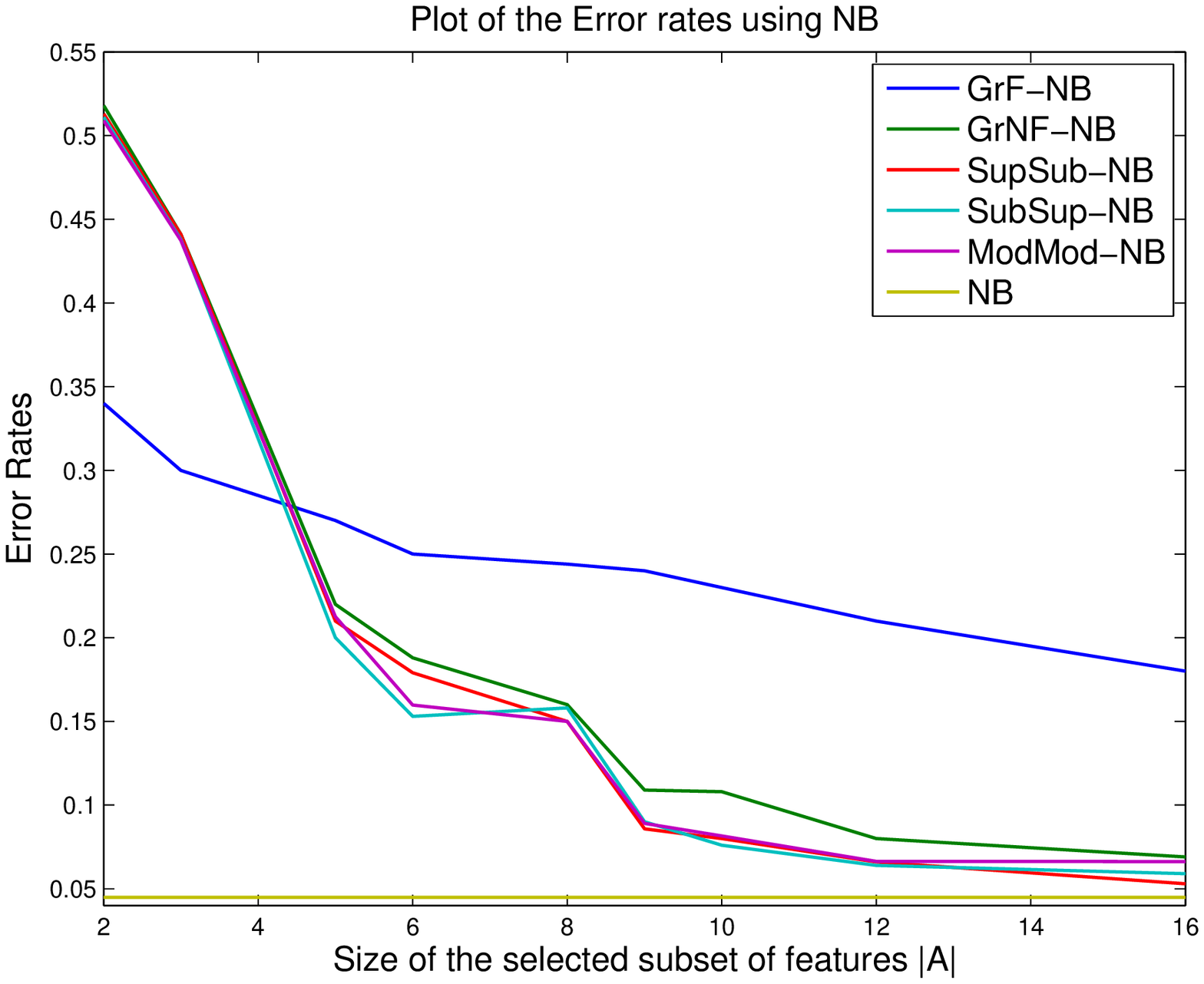}}
  ~ 
  \caption{Plot showing the accuracy rates vs.\ the number of features on the Mushroom data set.}
  \label{fig:mushrooms}
\end{figure}

\section{Experiments}


We test our algorithms on the feature subset selection problem in the
supervised setting. Given a set of features $X_V = \{X_1, X_2, \cdots,
X_{|V|}\}$, we try to find a subset of these features $A$ which has
the most information from the original set $X_V$ about a class
variable $C$ under constraints on the size or cost of $A$. 
Normally the number of features $|V|$ is quite large and thus the
training and testing time depend on $|V|$. In many cases, however,
there is a strong correlation amongst features and not every feature
is novel. We can thus perform training and testing with a much smaller
number of features $|A|$ while obtaining (almost) the same error
rates. \looseness-1

The question is how to find the most representative set of features
$A$. The mutual information between the chosen set of features and the
target class $C$, $I(X_A; C)$, captures the relevance of the chosen
subset of features. In most cases the selected features are not
independent given the class $C$ so the na\"{\i}ve Bayes assumption is
not applicable, meaning this is not a pure submodular optimization
problem. As mentioned in Section~\ref{sec:introduction}, $I(X_A; C)$
can be exactly expressed as a difference between submodular functions
$H(X_A)$ and $H(X_A | C)$.


\subsection{Modular Cost Feature Selection}
In this subsection, we look at the problem of maximizing $I(X_A; C) -
\lambda |A|$, as a regularized feature subset selection problem. Note
that a mutual information $I(X_A; C)$ query can easily be estimated
from the data by just a single sweep through this data. Further we
have observed that using techniques such as Laplace smoothing helps to
improve mutual information estimates without increasing
computation. In these experiments, therefore, we estimate
the mutual information directly from the data and run our algorithms
to find the representative subset of features.

We compare our algorithms on two data sets, i.e., the Mushroom data
set~\cite{iba1988trading} and the Adult data
set~\cite{kohavi1996scaling} obtained
from~\cite{Frank+Asuncion:2010}. The Mushroom data set has 8124
examples with 112 features, while the Adult data set has 32,561
examples with 123 features. In our experiments we considered subsets
of features of sizes between 5\%-20\% of the total number of
features
by varying $\lambda$. 
We tested the following algorithms for the feature
subset selection problem. We considered two formulations of the mutual
information, one under na\"{\i}ve Bayes, where the conditional entropy
$H(X_A | C)$ can be written as $H(X_A | C) = \sum_{j \in A} H(X_i |
C)$
and another where we do not assume such factorization.
We call these two formulations {\em factored} and {\em non-factored}
respectively. We then considered the simple greedy algorithm, of
iteratively adding features at every step to the factored and non-factored mutual
information, which we call GrF and
GrNF respectively. Lastly, we use the new algorithms presented in this paper on the
non-factored mutual information.

We then compare the results of the greedy algorithms with those of the
three algorithms for this problem, using two pattern classifiers
based on either a linear kernel 
SVM (using~\cite{libsvm}) or a na\"{\i}ve Bayes (NB) classifier. We
call the results obtained from the supermodular-submodular heuristic
as ``SupSub'', the submodular-supermodular
procedure~\cite{narasimhanbilmes} as ``SubSup'', and the
modular-modular objective as ``ModMod.'' In the SubSup procedure, we
use the minimum norm point algorithm~\cite{fujishige2011submodular}
for submodular minimization, and in the SubSup procedure, we use the
optimal algorithm of~\cite{feldman2012optimal} for submodular
maximization. We observed that the three heuristics generally
outperformed the two greedy procedures, and also that GRF can perform
quite poorly, thus justifying our claim that the na\"{\i}ve Bayes
assumption can be quite poor. This also shows that although the greedy
algorithm in that case is optimal, the features are correlated given
the class and hence modeling it as a difference between submodular
functions gives the best results. We also observed that the SupSub and
ModMod procedures perform comparably to the SubSup procedure, while
the SubSup procedure is \emph{much} slower in practice. Comparing the
running times, the ModMod and the SupSub procedure are each a few
times slower then the greedy algorithm (ModMod is slower due computing
the modular semigradients), while the SubSup procedure is around 100
times slower.  The SubSup procedure is slower due to general
submodular function minimization which can be quite slow.\looseness-1

The results for the Mushroom data set are shown in
Figure~\ref{fig:mushrooms}. We performed a 10 fold cross-validation on
the entire data set and observed that when using all the features SVM
gave an accuracy rate of 99.6\% while the all-feature NB model
had an accuracy rate of 95.5\%.  The results for the Adult database are in Figure~\ref{fig:adult}. In
this case with the entire set of features the accuracy rate of SVM on
this data set is 83.9\% and NB is 82.3\%. 
\begin{figure}[h]
  \centering
  \subfloat[SVM]{\label{fig:adultsvm}\includegraphics[width=0.23\textwidth]{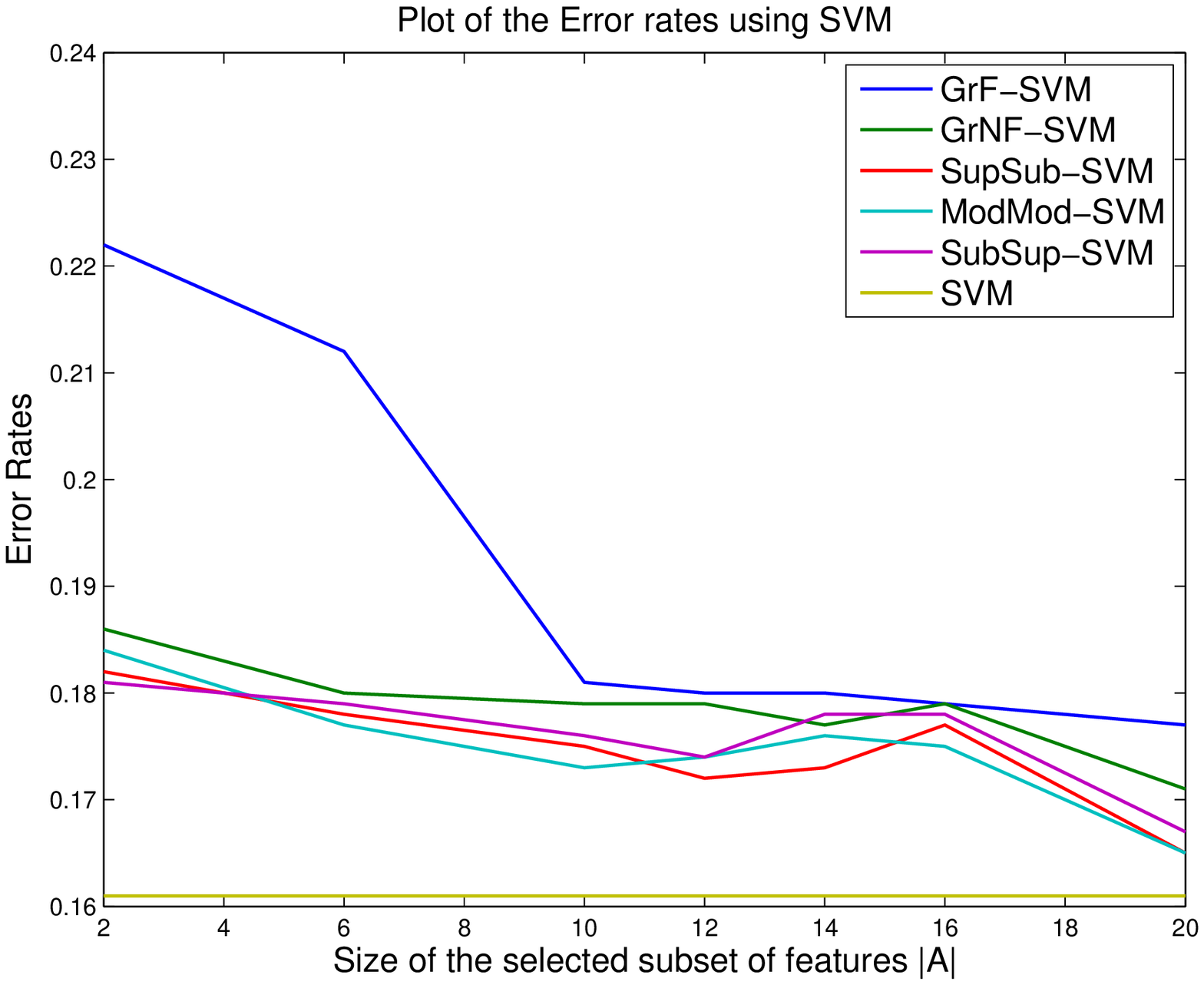}}
  ~ 
  \subfloat[NB]{\label{fig:adultnb}\includegraphics[width=0.23\textwidth]{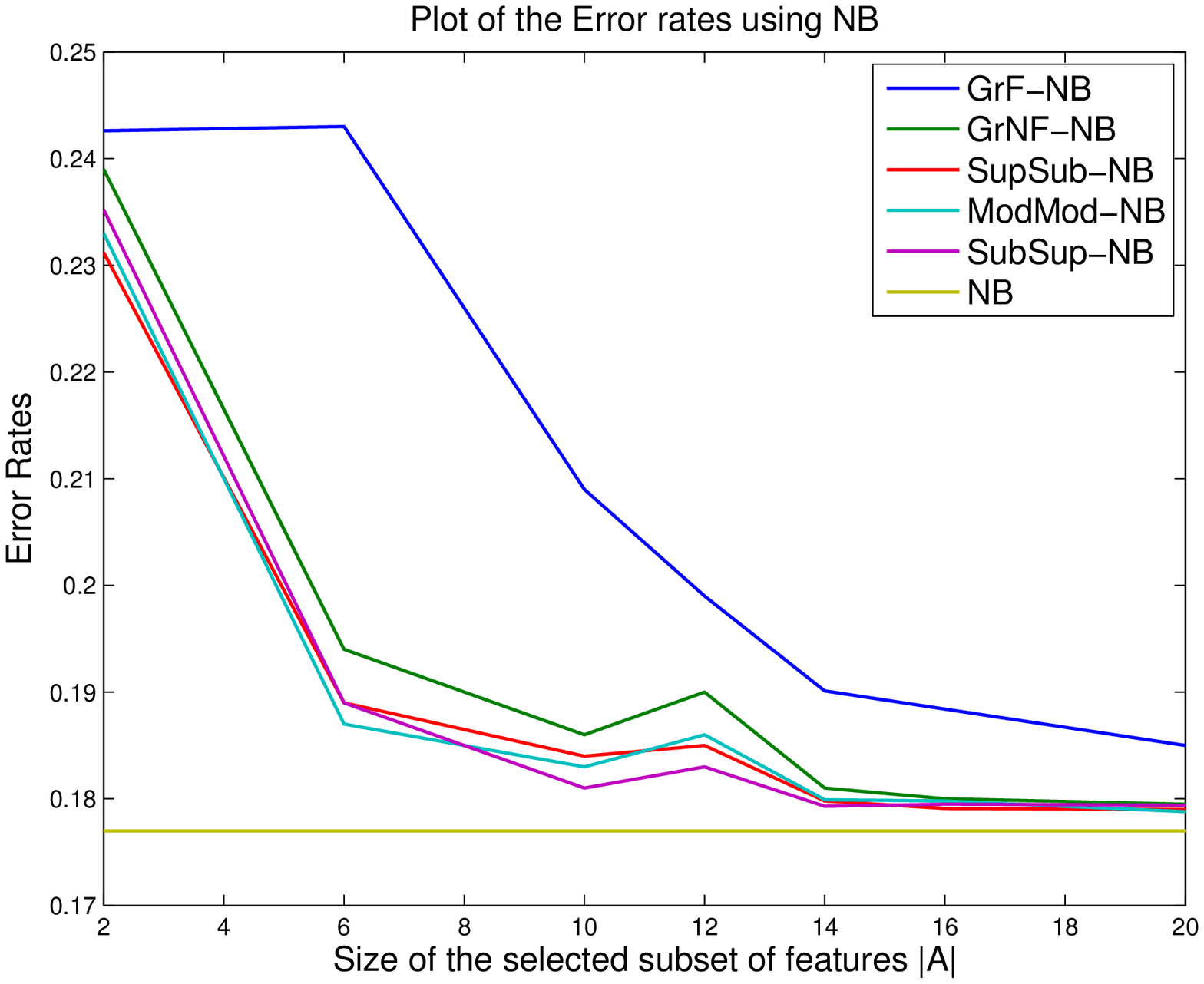}}
  ~ 
  \caption{Plot showing the accuracy rates vs.\ the number of features on the Adult data set.}
  \label{fig:adult}
\end{figure}

In the mushroom data, the SVM classifier significantly outperforms the
NB classifier and correspondingly GrF performs much worse than the
other algorithms. Also, in most cases the three algorithms outperform GrNF. In the adult data set, both the SVM and NB perform comparably
although SVM outperforms NB. However in this case also we observe that our algorithms generally outperform GrF and GrNF.



\subsection{Submodular cost feature selection}

\begin{figure}[h]
  \centering
  \subfloat[SVM]{\label{fig:mushroomssvmsub}\includegraphics[width=0.23\textwidth]{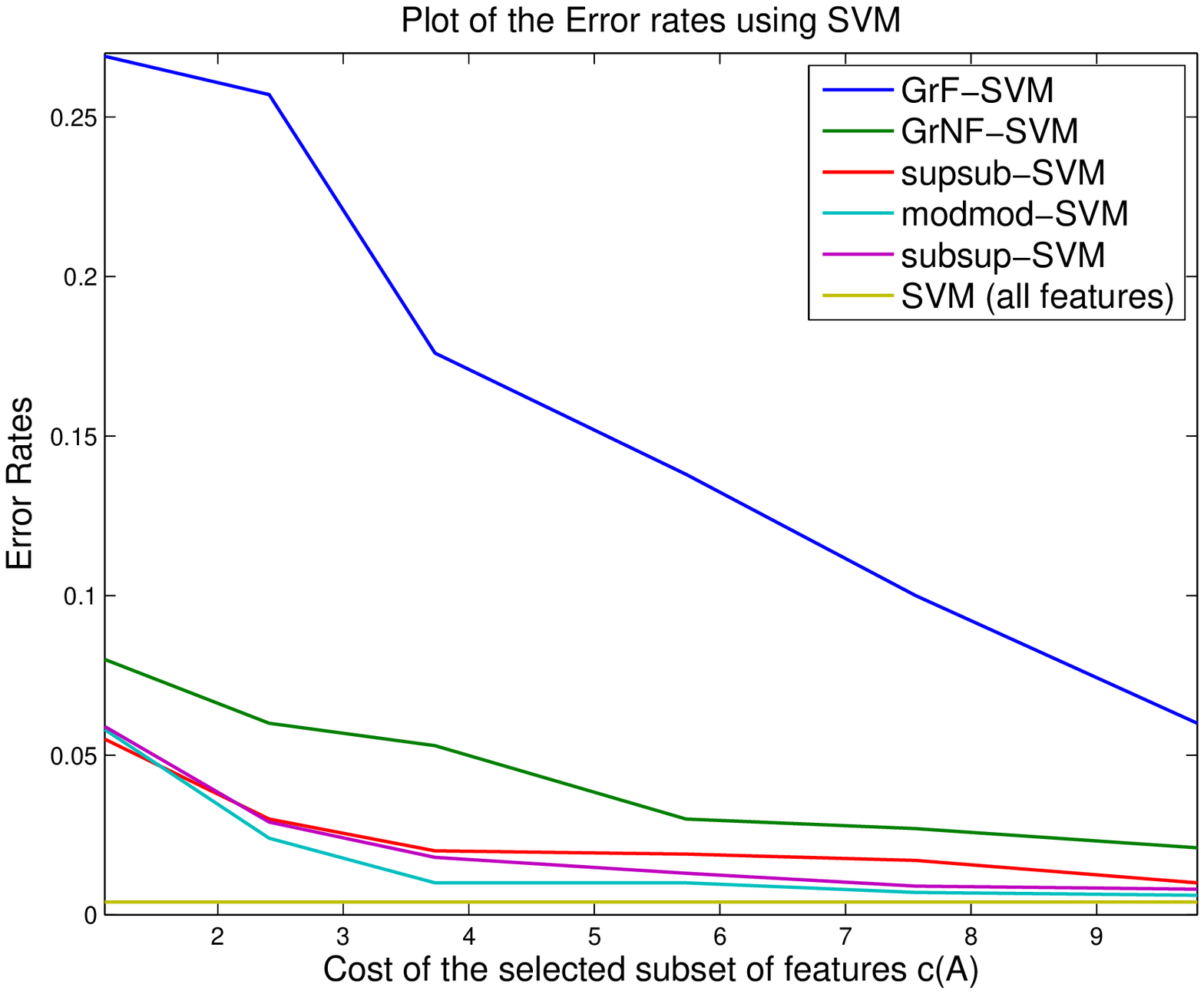}}
  ~ 
  \subfloat[NB]{\label{fig:mushroomsnbsub}\includegraphics[width=0.23\textwidth]{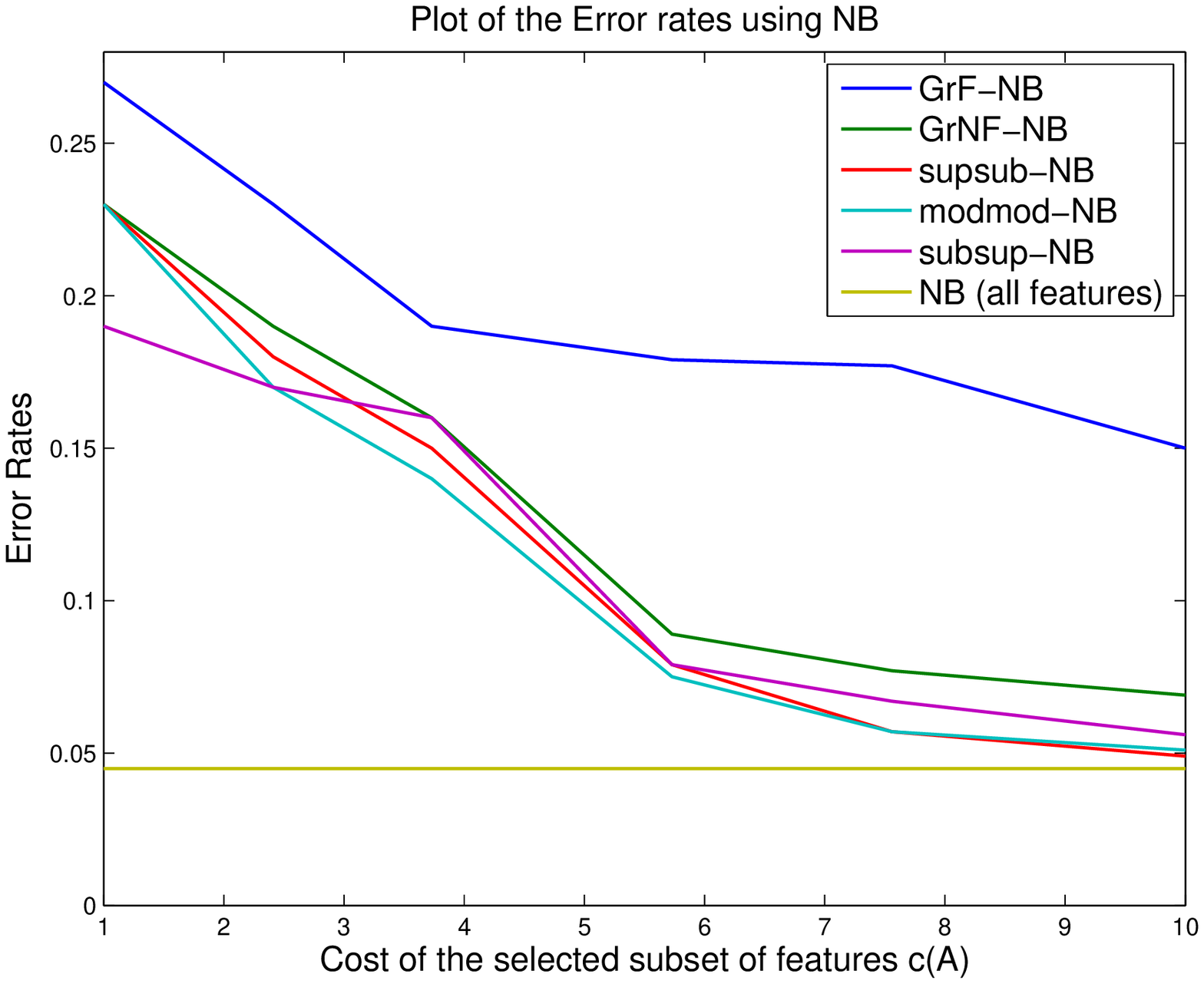}}
  ~ 
  \caption{Plot showing the accuracy rates vs.\ the cost of features for the Mushroom data set}
  \label{fig:mushroomsub}
\end{figure}
\begin{figure}[h]
  \centering
  \subfloat[SVM]{\label{fig:adultsvmsub}\includegraphics[width=0.23\textwidth]{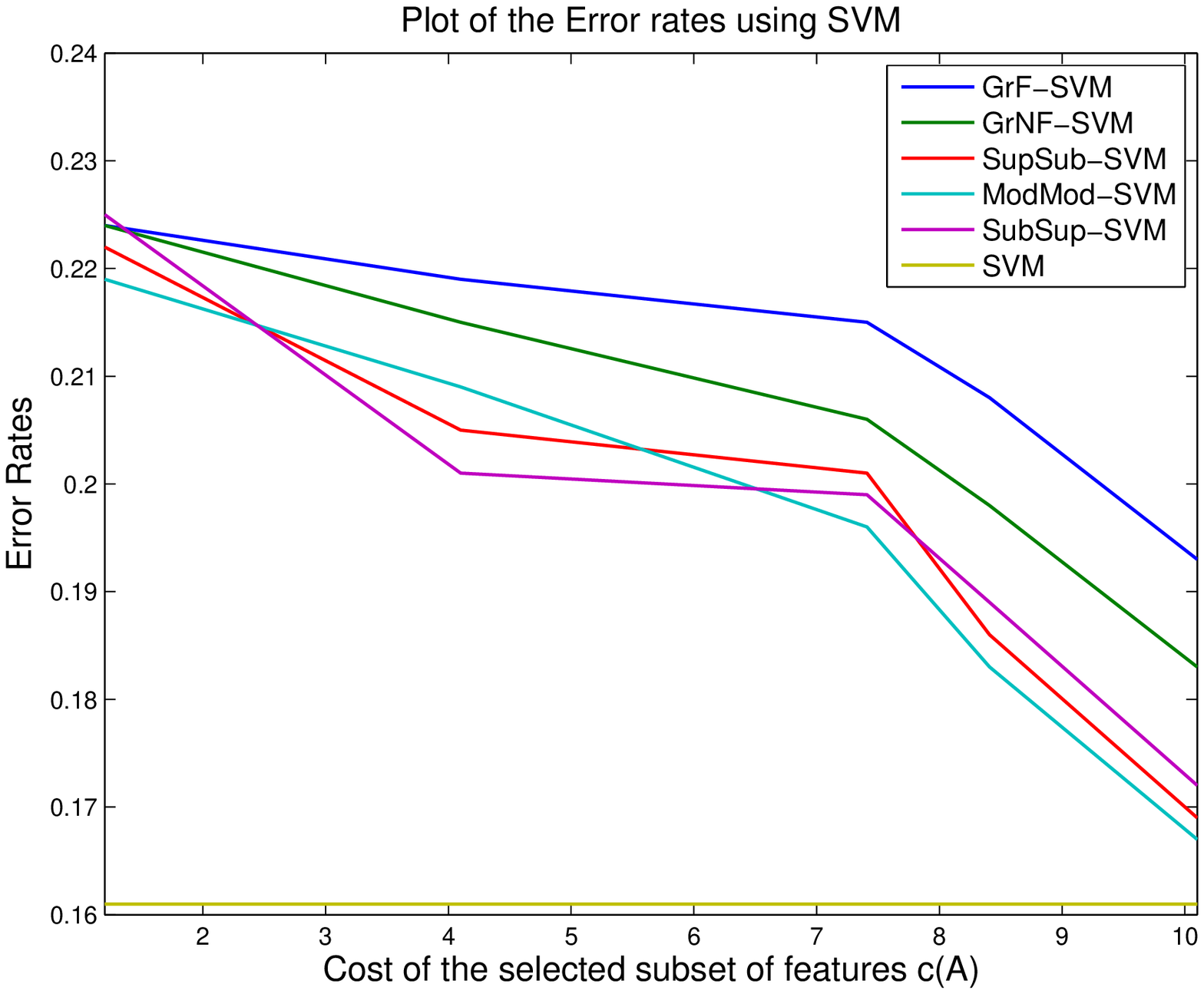}}
  ~ 
  \subfloat[NB]{\label{fig:adultnbsub}\includegraphics[width=0.23\textwidth]{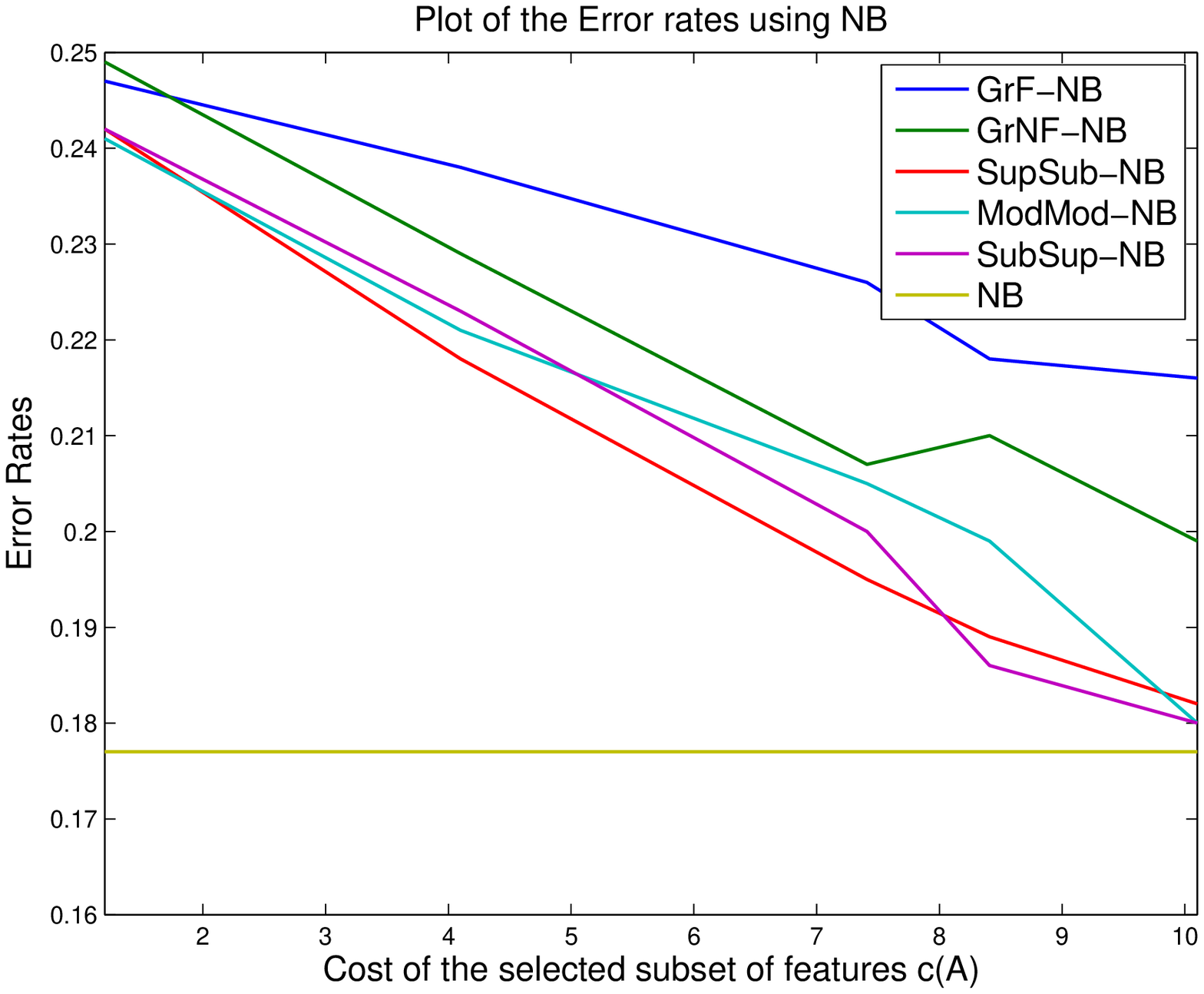}}
  ~ 
  \caption{Plot showing the accuracy rates vs.\ the cost of features for the Adult data set}
  \label{fig:adultsub}
\end{figure}

We perform synthetic experiments for the feature subset selection
problem under submodular costs. The cost model we consider is $c(A) =
\sum_i \sqrt{m(A \cap S_i)}$. We partitioned $V$ into sets $\{ S_i
\}_i$ and chose the modular function $m$ randomly.  In this set of
experiments, we compare the accuracy of the classifiers vs.\ the {\em
  cost} associated with the choice of features for the algorithms.
Recall, with simple (modular) cardinality costs the greedy algorithms
performed decently in comparison to our algorithms in the adult data
set, where the NB assumption is reasonable. However with submodular
costs, the objective is no longer submodular even under the NB
assumption and thus the greedy algorithms perform much worse.  This is
unsurprising since the greedy algorithm is approximately optimal only
for monotone submodular functions. This is even more strongly evident
from the results of the mushrooms data-set
(Figure~\ref{fig:mushroomsub})

\section{Discussion}
\label{sec:conclusions}

We have introduced new algorithms for optimizing the difference
between two submodular functions, provided new theoretical
understanding that provides some justification for heuristics, have
outlined applications that can make use of our procedures, and have
tested in the case of feature selection with modular and submodular
cost features.  Our new ModMod procedure is fast at each iteration and
experimentally
does about as well as the SupSub and SubSup
procedures. The ModMod procedure, moreover, can also be used under
various combinatorial constraints, and therefore the
ModMod procedure may hold the greatest promise as a practical heuristic.
An alternative approach, not yet evaluated, would be
to try the convex-concave procedure~\cite{yuille2002concave} on the
\lovasz{} extensions of $f$ and $g$ since subgradients in such case
are so easy to obtain.\looseness-1

{\bf Acknowledgments:} We thank Andrew Guillory, Manas Joglekar,
Stefanie Jegelka, and the rest of the submodular group at UW for
discussions.
This material is based upon work supported by the National Science
Foundation under Grant No. (IIS-1162606), and is also supported by a
Google, a Microsoft, and an Intel research award.
\looseness-1

\bibliographystyle{plainnat}
\bibliography{../Combined_bib/submod}
\end{document}